\documentclass[sigconf, authorversion, nonacm]{acmart}

\usepackage[utf8]{inputenc}
\usepackage{xcolor}
\usepackage{amsmath}
\usepackage{amsthm}
\usepackage{graphicx}
\usepackage[ruled]{algorithm2e}
\usepackage{float}                    
\usepackage{framed}
\usepackage{mdframed}
\usepackage{enumitem}

\newtheorem{definition}{Definition}
\newtheorem{theorem}{Theorem}
\newfloat{protocol}{tbp}{lop}
\floatname{protocol}{Protocol}
\newfloat{functionality}{tbp}{lop}
\floatname{functionality}{Functionality}

\renewcommand{\paragraph}[1]{\vspace*{2pt}\noindent\textbf{#1}}

\begin{document}
\title{Distributed HDMM: Scalable, Distributed, Accurate, and Differentially Private Query Workloads without a Trusted Curator}


\author{Ratang Sedimo}
\affiliation{\institution{University of Vermont}\country{}}
\email{rsedimo@uvm.edu}
\author{Ivoline C. Ngong}
\affiliation{\institution{University of Vermont}\country{}}
\email{kngongiv@uvm.edu}
\author{Jami Lashua}
\affiliation{\institution{University of Vermont}\country{}}
\email{jami.lashua@uvm.edu}
\author{Joseph P. Near}
\affiliation{\institution{University of Vermont}\country{}}
\email{jnear@uvm.edu}

\begin{abstract}
We present the Distributed High-Dimensional Matrix Mechanism (Distributed HDMM), a protocol for answering workloads of linear queries on distributed data that provides the accuracy of central-model HDMM without a trusted curator. Distributed HDMM leverages a secure aggregation protocol to evaluate HDMM on distributed data, and is secure in the context of a malicious aggregator and malicious clients (assuming an honest majority). Our preliminary empirical evaluation shows that Distributed HDMM can run on realistic datasets and workloads with thousands of clients in less than one minute.
\end{abstract}

\settopmatter{printfolios=true}
\maketitle




\section{Introduction}

Institutions like the U.S.\ Census Bureau and Medicare regularly release summary statistics 
about individuals, such as population statistics cross-tabulated by demographic attributes~\cite{census2010summaryfile, onthemapwebtool} 
and hospital discharge tables organized by medical condition and patient characteristics~\cite{hcupnet}. 
While such statistics are invaluable for both public policy and academic research, they can 
also reveal sensitive information---especially when combined with other publicly available 
datasets~\cite{haney2017utility,machanavajjhala2008privacy,vaidya2013identifying}. Over the last two decades, 
\emph{differential privacy} (DP)~\cite{dwork2006calibrating,dwork2014algorithmic,li2014data} 
has emerged as the premier framework for controlling privacy leakage in these kinds of data releases. 
It provides strong, formal, and quantifiable guarantees that remain valid even if an adversary 
possesses substantial auxiliary information.

Within differential privacy, a wide variety of mechanisms have been proposed to answer 
\emph{workloads of linear queries} under the central model. One notable example is the 
High-Dimensional Matrix Mechanism (HDMM)~\cite{mckenna2018optimizing}, which, for a given 
set of predicate-count or linear queries, can add a carefully calibrated amount of noise 
and exploit structural correlations across queries to minimize the overall error. 
Such \emph{central-model} mechanisms are capable of achieving optimal or near-optimal 
accuracy~\cite{acs2012differentially,barak2007privacy,ding2011differentially,hay2010boosting,li2015matrix,qardaji2014priview} 
and have proven effective on a variety of realistic workloads 
ranging from low- to high-dimensional data~\cite{li2013optimal,xiao2011differential,xiao2014dpcube}.

However, central-model mechanisms require a \emph{trusted curator} to collect and hold 
the raw data. In many settings, such as large-scale multi-institution collaborations or 
federated learning scenarios, there may be no single entity that the data holders fully trust. 
One alternative is \emph{local differential privacy} (LDP)~\cite{kasiviswanathan2011can,duchi2013local,erlingsson2014rappor,bassily2015local}, which 
eliminates the need for a centralized trusted party by having each data holder locally add 
noise to its own data. Yet, the local model typically introduces higher variance in the estimates 
and generally suffers from significantly poorer accuracy than central-model DP.

A growing body of work aims to bridge the gap between central and local models by leveraging 
secure computation protocols, such as secure aggregation or other forms of multi-party 
computation (MPC)~\cite{evans2018pragmatic}. These protocols allow a set of participants 
to compute aggregate statistics without revealing the raw data to any single party, thus 
eliminating the need for a fully trusted curator. At the same time, when used to implement differentially private algorithms, they can achieve accuracy 
comparable to that of central-model mechanisms~\cite{bater2018shrinkwrap, roy2020crypt, wang2022incshrink, roth2019honeycrisp, roth2020orchard}.

Motivated by the need to produce high-accuracy query answers \emph{without} relying on a 
fully trusted data curator, we propose the \emph{Distributed HDMM mechanism}---a distributed 
differential privacy protocol that preserves the advantages of central-model HDMM while 
operating in a secure multi-party setting.

Our approach combines the high accuracy of the central-model HDMM mechanism with the privacy and security properties of a distributed protocol. At a high level, Distributed HDMM enables each client to compute a noisy local contribution to a shared query workload, without ever transmitting raw data. Specifically, the server first computes a strategy matrix optimized for the given workload using HDMM, and broadcasts this matrix to all clients. Each client then encodes its local data into a vector representation, applies the strategy matrix to obtain a measurement, adds carefully calibrated discrete Gaussian noise, and finally encodes the result into a finite field. These encoded measurements are then securely aggregated using a communication-efficient secure aggregation protocol~\cite{bell2020secure}—ensuring that the server learns only the noisy sum of client measurements and nothing else. The server then decodes this aggregated sum, applies the HDMM postprocessing step (inverting the strategy matrix), and releases a differentially private answer to the original workload. This protocol maintains the algebraic structure of HDMM while ensuring that raw data never leaves client devices, and that even a malicious server colluding with clients cannot breach individual privacy—provided an honest majority assumption holds.

In contrast to prior distributed approaches that either introduce substantial error (as in local DP) or require additional trust assumptions (as in shuffle DP), Distributed HDMM directly mirrors the workflow of central HDMM while operating entirely in a federated setting. The key technical ingredients are the encoding step, which adapts real-valued measurements for secure aggregation, and the careful calibration of discrete Gaussian noise to account for potentially corrupted clients. Together, these design choices ensure that our protocol inherits the accuracy advantages of HDMM while providing end-to-end privacy and robustness guarantees in distributed deployments.

We perform an empirical evaluation of Distributed HDMM using workloads derived from the U.S.\ Census SF1 and Adult datasets, simulating deployments with thousands of clients to measure both runtime and accuracy. Our experiments examine the effects of client and server bandwidth limitations, network latency, and client dropouts. Results demonstrate that Distributed HDMM scales efficiently, introduces minimal computation and communication overhead, and achieves utility similar to central-model HDMM, while substantially outperforming local and shuffle-model baselines. We release our implementation as open source.\footnote{\url{https://github.com/uvm-plaid/distributed-hdmm}}

\paragraph{Contributions.}
In summary, our contributions are:
%
\begin{itemize}[leftmargin=14pt, itemsep=2pt]
\item We develop Distributed HDMM, a new distributed DP mechanism that extends the central-model HDMM to federated settings by leveraging secure aggregation, achieving near-central accuracy without a trusted curator.
\item We present security and privacy proofs for both the semi-honest and malicious threat models
\item We implement and evaluate Distributed HDMM to demonstrate that it scales to thousands of clients, runs in under a minute, and produces utility similar to central-model HDMM
\end{itemize}





\noindent

\paragraph{Paper Overview.} 
The remainder of the paper is organized as follows. Section~2 reviews background on differential privacy, the High-Dimensional Matrix Mechanism (HDMM), and secure aggregation as the cryptographic primitive underlying our distributed design. Section~3 introduces the core design of Distributed HDMM, describing how we adapt HDMM to a federated setting through vectorization, encoding, secure aggregation, and decoding. Section~4 presents the security analysis, covering semi-honest and malicious adversaries and discussing extensions with zero-knowledge input validation. Section~5 reports our empirical evaluation on U.S.\ Census SF1 and Adult datasets, demonstrating that Distributed HDMM scales to thousands of clients with minimal overhead and achieves accuracy close to central-model HDMM, while outperforming local and shuffle baselines. Section~6 concludes with a summary of contributions, practical deployment considerations, and directions for future work.

\section{Background}

\paragraph{Differential privacy.}
Differential privacy~\cite{dwork2006calibrating, dwork2014algorithmic} is a {formal privacy definition} that bounds the effect any single individual can have on the outcome of an analysis. Formally:
\begin{definition}[Differential privacy]
A {mechanism} $\mathcal{M}$ satisfies $(\epsilon, \delta)$-differential privacy if for all neighboring databases $x, x' \in \mathcal{D}$, and for all possible sets of outcomes $S$:
\[Pr[\mathcal{M}(x) \subseteq S] \leq e^\epsilon \Pr[\mathcal{M}(y) \subseteq S] + \delta\]
\end{definition}
Two databases are considered neighboring if they differ in one person's data.
In this work, we leverage a variant of differential privacy called \emph{zero-concentrated differential privacy} (zCDP)~\cite{bun2016concentrated}:
\begin{definition}[Zero-concentrated differential privacy (zCDP) {\cite{bun2016concentrated}}]
A randomized mechanism $\mathcal{M}$ satisfies $\rho$-zCDP if for all $\alpha\in(1,\infty)$ and all neighboring databases $x$ and $x'$,
\[
D_\alpha\!\big(\,\mathcal{M}(x)\,\big\|\,\mathcal{M}(x')\,\big)\;\le\;\rho\,\alpha,
\]
where $D_\alpha$ is the order-$\alpha$ Rényi divergence.
\end{definition}
To achieve zCDP, we will add random noise calibrated to the \emph{sensitivity} of the function whose output we would like to release. Sensitivity is defined as the maximum change in a function's output when one record differs between two neighboring databases. Sensitivity for vector-valued functions can be defined in terms of $L_1$ or $L_2$ norms:
\begin{definition}[L$_2$-sensitivity]
For a (vector-valued) query $q:\mathcal{X}^n\to\mathbb{R}^d$, the (global) L$_2$-sensitivity is
\[
\Delta_2(q)\;=\;\sup_{\text{neighbors }x\sim x'}\;\|\,q(x)-q(x')\,\|_2\, .
\]
\end{definition}
The most common mechanism for zCDP adds Gaussian noise calibrated to the privacy parameter $\rho$ and the $L_2$ sensitivity $\Delta_2$:
\begin{definition}[Gaussian mechanism for zCDP]
Let $q:\mathcal{D}\to\mathbb{R}^d$ have L$_2$-sensitivity $\Delta_2$. For an input database $x$, the Gaussian mechanism releases:
\[
q(x) + \mathcal{N}(0, \sigma^2)\text{ where }\sigma^2 = \frac{\Delta_2^2}{2\rho}
\]
\end{definition}
\begin{proposition}[Gaussian mechanism satisfies zCDP~\cite{bun2016concentrated}]
The Gaussian mechanism satisfies $\rho$-zCDP.
\end{proposition}






Like other variants of differential privacy, zCDP is closed under composition and post-processing, meaning that we can bound the total privacy cost of multiple uses of a mechanism, and that it is not possible to undo the privacy protection in a post-processing step.
\begin{lemma}[Composition~\cite{bun2016concentrated}]
If $\mathcal{M}_1$ satisfies $\rho_1$-zCDP and (possibly adaptively) $\mathcal{M}_2$ satisfies $\rho_2$-zCDP, then the joint mechanism $(\mathcal{M}_1,\mathcal{M}_2)$ satisfies $(\rho_1+\rho_2)$-zCDP.
\end{lemma}
\begin{lemma}[Post-processing~\cite{bun2016concentrated}]
If $\mathcal{M}$ satisfies $\rho$-zCDP and $f$ is any (possibly randomized) mapping, then $f\!\circ\!\mathcal{M}$ satisfies $\rho$-zCDP.
\end{lemma}
Finally, zCDP implies $(\epsilon, \delta)$-DP, and a $\rho$-zCDP guarantee can be converted into an $(\epsilon, \delta)$-DP guarantee (with some loss in precision):
\begin{proposition}[Conversion to $(\varepsilon,\delta)$-DP~\cite{bun2016concentrated}]
If $\mathcal{M}$ satisfies $\rho$-zCDP, then for any $\delta>0$, $\mathcal{M}$ is $(\varepsilon,\delta)$-DP with
\[
\varepsilon\;=\;\rho\;+\;2\sqrt{\rho\,\ln(1/\delta)}\, .
\]
\end{proposition}


\begin{algorithm}
\SetKwInOut{Input}{Input}
\SetKwInOut{Output}{Output}
\Input{Database $D \in \mathcal{D}$, query workload matrix $W$, privacy parameter $\rho$.}
\Output{A $\rho$-zCDP private approximate answer for $W$ on $D$.}

$A \gets \textsf{optimize}(W)$\tcp*{compute strategy matrix}
$x \gets \textsf{vectorize}(D)$\tcp*{vectorize database}
$M \gets Ax$\tcp*{compute measurement}
$\hat{M} \gets M + \mathcal{N}\Big(\frac{\Delta(A)^2}{2 \rho}\Big)$\tcp*{add noise}
$a \gets A^{-1}\hat{M}$\tcp*{reconstruct workload answer}
$\Return\;\; a$
\caption{High-Dimensional Matrix Mechanism (HDMM)~\cite{mckenna2018optimizing}.}
\label{alg:hdmm}
\end{algorithm}

\paragraph{HDMM.}
The High-Dimensional Matrix Mechanism~\cite{mckenna2018optimizing}, summarized in Algorithm~\ref{alg:hdmm}, answers a workload of linear queries with differential privacy. It works by representing the queries as a matrix and the database as a vector, then computing the matrix-vector product and adding Laplace noise to find differentially private query answers. However, HDMM makes several important optimizations that result in optimal accuracy. First, to optimize accuracy, HDMM uses an optimized \emph{strategy matrix} in place of the original query workload, then reconstructs the workload's answers from the result. Second, to scale to high-dimensional data and large query workloads, HDMM proposes an efficient \emph{implicit representation} for both the workload and the strategy matrix. HDMM represents the state-of-the-art in answering query workloads with high accuracy, but it is designed for the central model of differential privacy, and requires the data to be collected in one place by a data curator.

\paragraph{Secure aggregation.}
\emph{Secure aggregation} is a class of secure multiparty computation protocols for summing vectors, originally designed for federated machine learning~\cite{mcmahan2017communication, bonawitz2019towards, kairouz2019advances}. These protocols require only a few rounds of communication and scale much better than general-purpose multiparty computation protocols~\cite{bonawitz2017practical, bell2020secure, bell2022acorn, stevens2022secret, so2021turbo, kadhe2020fastsecagg, fereidooni2021safelearn}, enabling deployment in large-scale privacy-preserving federated learning systems~\cite{truex2019hybrid, sav2020poseidon}. Recent approaches~\cite{bell2020secure, bell2022acorn, stevens2022secret} scale to thousands or millions of participants and support both semi-honest and malicious adversaries. Formally, the secure aggregation functionality (Functionality~\ref{func:secagg}) takes as input vectors from all clients and outputs only their sum to the server, without revealing any individual contribution.

\begin{functionality}
\begin{mdframed}[align=center, userdefinedwidth=.45\textwidth]
\textbf{Parameters:}
\begin{itemize}[itemsep=0pt, topsep=0pt, leftmargin=12pt]
\item Parties: one server $S$ and $n$ clients $c_1 \dots c_n$.
\item Each client $c_i$ holds a length-$k$ vector of field elements $x_i \in \mathbb{F}_p^k$.
\end{itemize}
\medskip    

\textbf{Functionality:}
\begin{enumerate}[itemsep=0pt, topsep=0pt, leftmargin=16pt]
\item Each client $c_i$ sends $x_i$ to $\mathcal{F}_\text{agg}$
\item $\mathcal{F}_\text{agg}$ sends the sum $\sum_{i=1}^n x_i$ to $S$
\end{enumerate}

\end{mdframed}
  \caption{Secure aggregation functionality $\mathcal{F}_{\text{agg}}$.}
  \label{func:secagg}
\end{functionality}

\section{Threat Models}
\label{sec:threat-models}

We consider two threat models: semi-honest, and a variant of malicious security without correctness associated with secure aggregation protocols. Both settings consider a static set of parties \emph{corrupted} by the adversary: $C \subset \{S\} \cup \{c_1, \dots, c_n\}$. Our protocols are secure even when both the server $S$ and a $\theta$ fraction of the clients are corrupted.

\paragraph{Semi-honest security.}
In the semi-honest (also called honest-but-curious) threat model, all parties follow the protocol, and do not change their inputs or outputs. The adversary cannot affect the correctness of the result in this setting, since all parties follow the protocol. However, the adversary may try to learn information about the honest parties' inputs by observing messages received by the corrupt parties. Single-server secure aggregation protocols generally provide semi-honest security when both the server and a fraction of clients are corrupted. We prove semi-honest security for Distributed HDMM in Section~\ref{sec:security--privacy}.

\paragraph{Malicious security.}
In the malicious (also called active) threat model, each corrupted party can deviate arbitrarily from the protocol, including by changing their inputs, their message contents, and their outputs.
Single-server secure aggregation protocols~\cite{bell2022acorn, bonawitz2017practical, bell2020secure} typically do not ensure correctness of their outputs in the presence of a corrupted server in the malicious model, since the server is the single source of the final result and is allowed to output an arbitrary value when corrupted. These protocols do provide confidentiality for the clients' inputs, even in the presence of a malicious adversary---even when both the server and a fraction of the clients deviate from the protocol, the adversary cannot learn more about the honest clients' inputs than is revealed by the total sum.

We prove the same kind of malicious security for Distributed HDMM in Section~\ref{sec:security--privacy}. Since our approach relies on single-server secure aggregation protocols, it cannot guarantee correctness of the output when the server is corrupted, but it does provide confidentiality for honest clients' inputs.

This variant of malicious security is weaker than the traditional definition, which does include correctness of the output. However, this threat model is often considered to be a good match for practical deployments, in which a large company may operate an aggregation server to collect statistics about private data from a large number of customers. In this context, the company is highly incentivized to produce the correct result, since their goal is to use the collected statistics for business purposes; the customers, in contrast, primarily want to be protected from malicious behavior by the company that may violate their privacy.

\paragraph{Single- vs. multi-server aggregation.}
Highly-efficient secure aggregation protocols exist for the multi-server setting, where multiple non-colluding servers collaboratively compute the aggregate result. In practice, however, deployed applications of secure aggregation protocols typically involve collection of data by a single organization, and identifying additional organizations which verifiably do not collude with the first one is difficult. Our presentation and experiments focus on single-server setting, since it is more challenging and typically more applicable than the multi-server setting, but our approach extends in a trivial way to multi-server aggregation protocols.

\paragraph{Input validation.}
Malicious clients may have no incentive to provide correct inputs to the protocol, and may destroy the final output by providing garbage inputs.
Our approach is compatible with existing work that uses a zero-knowledge (zk) proof framework to prove that all clients' inputs are within a reasonable range.
Our malicious-secure protocol is compatible with several existing solutions for client input validation, including those proposed by ACORN~\cite{bell2022acorn} and EiFFeL~\cite{roy2022eiffel}. Employing any of these methods ensures that we can guarantee correctness of client inputs, in case some clients behave maliciously.

\section{Distributed HDMM}

\begin{table}
  \centering
  \begin{tabular}{c l}
    \hline
    \textbf{Parameter} & \textbf{Description}\\
    \hline
    $S$ & Server\\
    $n$ & Number of clients\\
    $c_1, \dots, c_n$ & Clients\\
    $\theta$ & Upper bound on fraction of corrupted clients\\
    $\rho$ & Privacy parameter (zCDP)\\
    $\gamma$ & Scaling parameter\\
    \hline
  \end{tabular}
  \caption{List of global parameters for the Distributed HDMM approach.}
  \label{tab:parameters}
\end{table}

\begin{figure*}
    \centering
    \includegraphics[width = 0.9\textwidth]{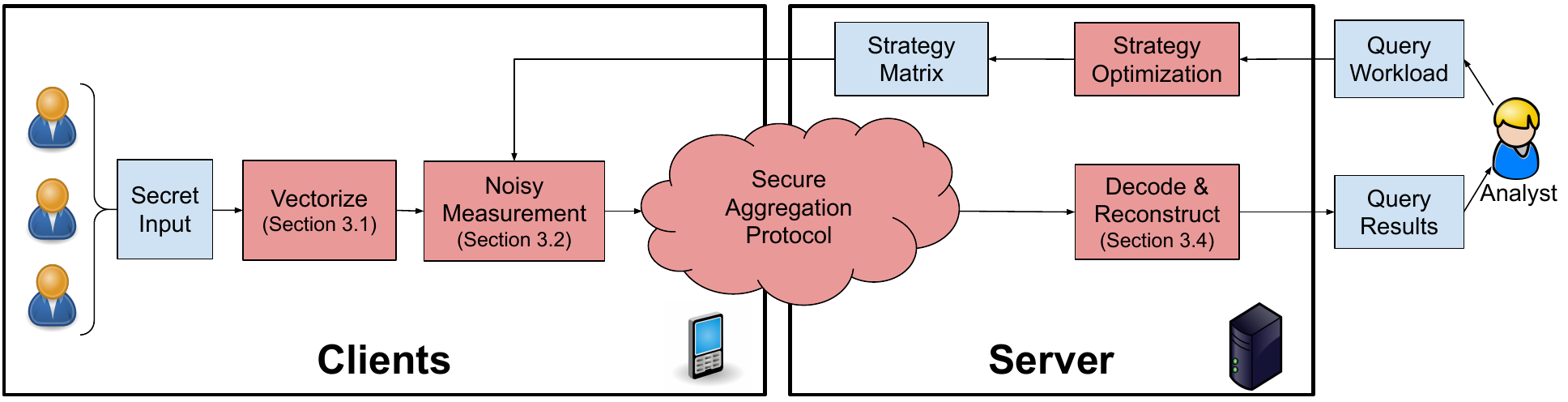}
    \caption{Overview of the Distributed HDMM approach.}
\end{figure*}

\begin{functionality}
\begin{mdframed}[align=center, userdefinedwidth=.45\textwidth]
\textbf{Parameters:}
\begin{itemize}[itemsep=0pt, topsep=0pt, leftmargin=12pt]
\item Parties: one server $S$ and $n$ clients $c_1 \dots c_n$.
\item Each client $c_i$ holds a multiset of tuples $I_i \in \mathcal{D}$.
\item The server $S$ holds a query workload matrix $W$.
\end{itemize}
\medskip    

\textbf{Functionality:}
\begin{enumerate}[itemsep=0pt, topsep=0pt, leftmargin=16pt]
\item Each client $c_i$ sends $I_i$ to $\mathcal{F}_\text{HDMM}$
\item $S$ sends $W$ to $\mathcal{F}_\text{HDMM}$
\item $\mathcal{F}_\text{HDMM}$ computes $D = \bigcup_{i=1}^n I_i$
\item $\mathcal{F}_\text{HDMM}$ computes the HDMM output $a$ according to Algorithm~\ref{alg:hdmm}, using $D$, $W$, and $\rho$
\item $\mathcal{F}_\text{HDMM}$ sends $a$ to $S$
\end{enumerate}

\end{mdframed}
  \caption{Distributed HDMM functionality $\mathcal{F}_{\text{HDMM}}$.}
  \label{func:hdmm}
\end{functionality}

\begin{protocol}
  \begin{mdframed}[align=center, userdefinedwidth=.45\textwidth]
\textbf{Parameters:}
\begin{itemize}[itemsep=0pt, topsep=0pt, leftmargin=10pt]
\item Parties: one server $S$ and $n$ clients $c_1 \dots c_n$.
\item Each client $c_i$ holds a multiset of tuples $I_i \in \mathcal{D}$.
\item The server $S$ holds a query workload matrix $W$.
\end{itemize}
\medskip    

\textbf{Output:} {A differentially private approximate answer $a$ for $W$ on the distributed data $x$. }
\medskip

{\textbf{Round 1:} The server $S$ computes the strategy matrix $A \in \mathbb{R}^{d \times k} =
  \textsf{optimize}(W)$ following HDMM, and broadcasts $A$ to the
  clients.}
\medskip

{\textbf{Round 2:} Each client $c_i$:}
\begin{enumerate}[topsep=1pt, itemsep=0pt, leftmargin=14pt]
\item Computes the $L_2$ sensitivity $\Delta_2$ of the strategy matrix $A$, following HDMM.
\item Computes the \emph{vector representation} of their data (\S\ref{sec:vectorization}): \\ $x_i \in \mathbb{R}^d = \textsf{vectorize}(I_i)$
\item Computes their \emph{measurement} using the matrix-vector product of $A$ and $v_i$, following HDMM: \\ $m_i \in \mathbb{R}^k = A x_i$
\item Computes their \emph{noisy measurement} (\S\ref{sec:encoding}): \\ $\hat{m_i} \in \mathbb{F}_p^k = \textsf{encode}(m_i, \Delta_2)$
\item Runs the aggregation protocol $\mathcal{F}_{\text{agg}}$ with the server $S$ and other clients, providing $\hat{m_i}$ as input (\S\ref{sec:aggregation}).
\end{enumerate}
\medskip

{\textbf{Round 3:} The server $S$:}
\begin{enumerate}[topsep=1pt, itemsep=0pt, leftmargin=14pt]
\item Receives $\hat{M} \in \mathbb{F}_p^k = \sum_{i} \hat{m_i}$ as the output of $\mathcal{F}_{\text{agg}}$.
\item Decodes the sum of measurements (\S\ref{sec:decoding}):\\ $\hat{M_d} \in \mathbb{R}^k = \textsf{decode}(\hat{M})$
\item Computes the approximate answer to $W$ using the inverse of the strategy matrix $A$, following HDMM: $a = A^{-1} \hat{M_d}$
\item Outputs $a$.
\end{enumerate}

\end{mdframed}
  \caption{Distributed HDMM Protocol $\prod_\text{HDMM}$.}
  \label{prot:hdmm}
\end{protocol}

The complete protocol for Distributed HDMM appears in Protocol~\ref{prot:hdmm}. The protocol involves $n$ clients $c_i$ and one server $S$. Client $c_i$ holds a set of tuples $I_i$. The protocol concludes by outputting a differentially private answer to the specified query workload.

The structure of the protocol follows the structure of centralized HDMM. The main difference is that each client computes a measurement using their own data, then adds noise to their measurement. The protocol aggregates these noisy measurements to arrive at the final result.




\subsection{Vectorization}
\label{sec:vectorization}

As in HDMM, we assume a single-table relational schema $R(A_1, \dots, A_d)$ where $attr(R)$ denotes the set of attributes of $R$, and that each attribute $A_i$ has a finite domain $dom(A_i)$. The full domain of $R$ is the product of the attributes' domains, and has size $|dom(R)| = \prod_i |dom(A_i)|$. HDMM represents an instance of the schema $R$ as a length-$d$ vector $x_I \in \mathbb{N}^d$. The vector is indexed by tuples $t \in dom(R)$ such that $x_I(t) = \sum_{t' \in dom(R)} \mathbb{I}[t = t']$---in other words, $x_I(t)$ counts the number of occurrences of the tuple $t$ in the original instance $I$. Dwork and Roth~\cite{dwork2014algorithmic} call this \emph{histogram representation} of the data.

In our setting, client $c_i$ holds a subset $I_i$ of the total instance, so $I_i \subset I$ and $\bigcup_i I_i = I$. Each client vectorizes their subset in exactly the same manner as in centralized HDMM, to obtain a vector $x_i \in \mathbb{N}^n$ that counts the number of occurrences of each tuple in the domain \emph{for that client only}. In the case where each client holds one tuple, each vector $x_i$ is a one-hot encoding of the tuple.
By construction, the sum of these vectors is equal to the vectorization of the original instance: $\sum_i x_i = x_I$.

\subsection{Encoding}
\label{sec:encoding}

The encoding algorithm encodes and adds noise to the client's measurement vector, outputting a vector of field elements for aggregation. This algorithm takes as an input, each client's response to the query and outputs a noisy encoded value representing each client's response.
The algorithm has to encode because the secure aggregation protocol $\mathcal{F}_{\text{agg}}$ requires its input to be vector of finite field elements.

The encoding algorithm (Algorithm~\ref{alg:encoding}) receives a measurement vector $v$ from the HDMM matrix, a scaling factor $\gamma$ and $L_2$ sensitivity $\Delta_2$. The HDMM measurement vector $v$, is scaled by $\gamma$ and truncated to an integer, and the product is perturbed by adding discrete Gaussian noise $\mathcal{N}_{\mathbb{Z}}\Big(\frac{\theta \gamma^2\Delta_2^2}{2 n \rho}\Big)$. The $\theta$ here represents the fraction of potentially corrupted parties. The $\theta$ scales up the noise to compensate for the potentially lost appropriate level of noise needed to guarantee differential privacy.

Following the addition of Gaussian noise, the perturbed matrix product, is a noisy vector of integers. Next, the algorithm transforms the integers into elements of the finite field $\mathbb{F}_p$, to ensure that our protocol can adequately handle both negative and positive numbers. Given a prime $p > 2$, we transform integers in $[-(p-1)/2, (p-1)/2]$ to field elements in $[0,p-1]$ by computing $\hat{v_s} \mod p$. Our implementation throws an error for inputs outside of this range, since they would yield incorrect results.

\begin{algorithm}
\SetKwInOut{Input}{Input}
\SetKwInOut{Output}{Output}
\Input{Measurement vector $v \in \mathbb{R}^k$, scaling factor $\gamma$, $L_2$ sensitivity $\Delta_2$.}
\Output{Encoded noisy measurement vector $\hat{v_s} \in \mathbb{F}_p^k$}
$v_s \gets \lfloor \gamma v \rfloor$\tcp*{scale, truncate encoded vector}
$\hat{v_s} \gets v_s + \mathcal{N}_{\mathbb{Z}}\Big(\frac{\gamma^2\Delta_2^2}{2 (1 - \theta) n \rho}\Big)$\tcp*{add discrete Gaussian noise}
$\Return\;\; \hat{v_s} \mod p$\tcp*{encode noisy value in $\mathbb{F}_p^k$}

\caption{Encoding procedure for a single client's noisy measurement.}
\label{alg:encoding}
\end{algorithm}



    



\subsection{Aggregation}
\label{sec:aggregation}





In our experiments, we instantiate $\mathcal{F}_{\text{agg}}$ with the Bell et al.~\cite{bell2020secure} protocol for secure aggregation. This protocol was designed for federated learning settings and includes both semi-honest and malicious-secure variants. The Bell protocol masks each client’s input vector using correlated randomness shared with other clients, such that the masks cancel out when all contributions are aggregated. As a result, the server learns only the sum of the noisy measurement vectors, and nothing about any individual client’s input. 


\subsection{Decoding}
\label{sec:decoding}

The decoding algorithm takes in results as a vector of field elements and decodes them into floating-point numbers. It is the inverse of the encode algorithm. The input is an encoded result vector $v_s \in \mathbb{F}_p^k$, and the output is a vector $v \in \mathbb{R}^k$.

\begin{algorithm}
\SetKwInOut{Input}{Input}
\SetKwInOut{Output}{Output}
\Input{An encoded vector $v_s \in \mathbb{F}_p^k$}
\Output{A decoded vector $v \in \mathbb{R}^k$}
$v_d \gets \textsf{DecodeInt}(v[i]) \textbf{ for } v[i] \in v_s$\tcp*{decode}
$\;\;\;\;\text{where } \textsf{DecodeInt}(v_i) = 
\begin{cases}
      x & \text{if }x \leq \frac{p-1}{2} \\
      x-p & \text{if }x > \frac{p-1}{2}
    \end{cases}\,$\\
$v \gets x / \gamma$ \tcp*{descale}
$\Return \;\; v$\;
\caption{Decoding procedure.}
\label{alg:decode}
\end{algorithm}

The decode algorithm (Algorithm~\ref{alg:decode}) performs this transformation. The first step is to map field elements to signed integers. The \textsf{DecodeInt} function maps field elements between 0 and $\frac{p-1}{2}$ to positive integers, and field elements between $\frac{p-1}{2}$ to negative integers; this is the inverse of the process used to encode signed integers as field elements. The second step is to transform signed integers into floating-point numbers by de-scaling to invert the process used in encoding, by dividing the decoded vector by $\gamma$.

\subsection{Complexity Analysis}

We now analyze the additional complexity of our distributed protocols relative to the central-model HDMM baseline. In the central setting, the server computes the query workload directly on the global dataset, while in the distributed setting clients must participate in secure aggregation and local encoding of their contributions. We break down the additional computational and communication costs for both server and clients, based on the use of the protocol due to Bell et al.~\cite{bell2020secure} to perform secure aggregation.

\paragraph{Computation cost.}
For each \emph{client}, the additional cost arises from secure aggregation as well as from preparing their contribution to the measurement. Each client must vectorize and encode its local measurement (a length-$k$ vector), at a cost of $O(k)$. The computation cost of secure aggregation for the client is $O(\log(n)^2 + k \log(n))$---so our approach scales linearly in the size of the strategy, but logarithmically in the number of clients. For the \emph{server}, the dominant additional cost relative to central-model HDMM is due to secure aggregation: $O(n \log(n)^2 + n k \log(n))$, so our protocol is linear in both the size of the strategy and the number of clients.

\paragraph{Communication cost.}
For each \emph{client}, the communication cost resulting from secure aggregation is $O(\log(n)^2 + k)$.
For the \emph{server}, the communication cost is $O(n \log(n)^2 + n k)$. In both cases, the additional communication overhead compared to local or shuffle model alternatives is logarithmic.

\section{Security \& Privacy Analysis}
\label{sec:security--privacy}



This section provides the formal proofs of security and privacy for Protocol~\ref{prot:hdmm}. In Section~\ref{sec:privacy_analysis}, we prove that Protocol~\ref{prot:hdmm} satisfies differential privacy. We prove security (confidentiality) for inputs in the context of a semi-honest adversary in Section~\ref{sec:semihonest_security}, and for a malicious adversary in Section~\ref{sec:maliciou_security}.


\subsection{Privacy Analysis}
\label{sec:privacy_analysis}

In Protocol~\ref{prot:hdmm}, each client adds a small amount of noise ($O(\frac{1}{n})$), which is not sufficient to ensure differential privacy. When the noisy measurements are summed, however, the noise samples add up to the correct noise for differential privacy. For samples from the continuous Gaussian, this result is immediate; Distributed HDMM uses discrete Gaussian noise, however, since the vector to be aggregated consist of field elements in $\mathbb{F}_p$. Fortunately, a sum of discrete Gaussian samples can also be used to satisfy differential privacy, and if the scaling factor $\gamma$ is large enough, then the guarantee is very nearly the same as in the continuous case.
\begin{lemma}[Distributed discrete Gaussian~\cite{kairouz2021distributed}]
  Let $\sigma \geq 1$ and $X_i \sim \mathcal{N}_\mathbb{Z}(0, \sigma^2)$
  independently for each $i$. Let $Z_n = \sum_{i=0}^n X_i$. An
  algorithm that adds $Z_n$ to a sensitivity-$\Delta$ query satisfies
  $\rho'$-zero concentrated differential privacy, for:
  \[ \rho' = \frac{\Delta^2}{2n\sigma^2} + 5 \sum_{k=1}^{n-1}e^{-4\pi^2\sigma^2\frac{k}{k+1}}\]
  \label{lem:discgauss}
\end{lemma}
\begin{theorem}
  The output of Protocol~\ref{prot:hdmm} satisfies $\rho'$-zero concentrated differential privacy for:
  \[ \rho' = \rho + \kappa\]
where
\[\kappa = 5 \sum_{k=1}^{n(1-\theta)-1}e^{-4\pi^2(\frac{\gamma^2\Delta_2^2}{2 (1 - \theta) n \rho})\frac{k}{k+1}}\]
  \label{thm:privacy}
\end{theorem}
\begin{proof}

Each client's measurement vector $\hat{m}_i$ (as computed by \textsf{encode}) has sensitivity $\gamma \Delta_2$ and has independent discrete Gaussian noise sampled from:
\[\mathcal{N}_{\mathbb{Z}}\Big(\frac{\gamma^2\Delta_2^2}{2 (1 - \theta) n \rho}\Big)\]
The server $S$ receives the sum $\hat{M} = \sum_i \hat{m}_i$. In the worst case, the server can subtract the noise of the $n \theta$ corrupted clients, so $n(1-\theta)$ noise samples remain. By Lemma~\ref{lem:discgauss}, the sum of the measurement vectors in the presence of $n\theta$ corrupted clients satisfies $\rho'$-zCDP, where:
\begin{align*}
\rho' &= \frac{\Delta^2}{2n(1-\theta)\sigma^2} + 5 \sum_{k=1}^{n(1-\theta)-1}e^{-4\pi^2\sigma^2\frac{k}{k+1}}\\
&= \frac{\Delta^2}{2n(1-\theta)\frac{\gamma^2\Delta_2^2}{2 (1 - \theta) n \rho}} + 5 \sum_{k=1}^{n(1-\theta)-1}e^{-4\pi^2(\frac{\gamma^2\Delta_2^2}{2 (1 - \theta) n \rho})\frac{k}{k+1}}\\
&= \frac{\Delta^2}{\frac{\gamma^2\Delta_2^2}{\rho}} + 5 \sum_{k=1}^{n(1-\theta)-1}e^{-4\pi^2(\frac{\gamma^2\Delta_2^2}{2 (1 - \theta) n \rho})\frac{k}{k+1}}\\
&= \rho + 5 \sum_{k=1}^{n(1-\theta)-1}e^{-4\pi^2(\frac{\gamma^2\Delta_2^2}{2 (1 - \theta) n \rho})\frac{k}{k+1}}\\
\end{align*}
\end{proof}
The $\kappa$ term of the privacy cost in Theorem~\ref{thm:privacy} corresponds to the ``extra cost'' of using discrete Gaussian noise samples rather than continuous ones. Fortunately, this extra term shrinks exponentially with the scaling parameter $\gamma$, so it can easily be made negligible by setting $\gamma$ large enough. For example, when $\rho = 0.1$ and $n=5000$, setting $\gamma = 100$ results in $\kappa = 9.39 \times 10^{-86}$. In our experiments, we set $\gamma = 1000$, which results in a $\kappa$ that is too small to calculate using 64-bit floating-point numbers.


\subsection{Semi-Honest Security}
\label{sec:semihonest_security}

In addition to satisfying differential privacy, Protocol~\ref{prot:hdmm} must not reveal anything new to any party except for the final result. This security result follows directly from the security of the secure aggregation protocol used. The clients communicate their private data only through the aggregation protocol $\mathcal{F}_{\text{agg}}$; if $\mathcal{F}_{\text{agg}}$ is secure, then no client learns anything about any other client's input, and the server learns only the sum of these inputs.

\begin{theorem}[Security of $\prod_\text{HDMM}$]
The protocol $\prod_\text{HDMM}$ (Protocol~\ref{prot:hdmm}) securely realizes the functionality $\mathcal{F}_\text{HDMM}$ (Functionality~\ref{func:hdmm}) in the $\mathcal{F}_\text{agg}$-hybrid model, in the presence of semi-honest adversaries.
\label{thm:semi-honest-security}
\end{theorem}
\begin{proof}
We show the existence of a polynomial-time simulator via a hybrid argument~\cite{evans2018pragmatic}. We assume that $U = \{c_1, \dots, c_n\}$ is the set of clients, $S$ is the server, and $C \subset U \cup \{S\}$ is the set of corrupted clients. $I_i$ is the input of client $c_i$ and $V_i$ is the view of client $c_i$. Let $\pi = \prod_\text{HDMM}$ and $\mathcal{F} = \mathcal{F}_\text{HDMM}$. We need to show the existence of a simulator $\textsf{Sim}$ such that:
\[ \textsf{Real}_\pi(\kappa; I) \equiv
\textsf{Ideal}_{\mathcal{F}, \textsf{Sim}}(\kappa; \{I_i \mid c_i \in C \}) \]
Where both \textsf{Real} and \textsf{Ideal} output views of the corrupt parties $\{V_i \mid c_i \in C\}$. The view $V_i$ of $c_i$ contains $c_i$'s private input, its random tape, and all messages received during the protocol.

We proceed via a hybrid argument, beginning with the real protocol and ending at the simulator. At each step, we argue that the views produced by the new hybrid are indistinguishable from those produced by the previous one.
\begin{description}[leftmargin=18pt]
\item[$\textsf{Hyb}_1$] This hybrid is identical to $\textsf{Real}_\pi(\kappa; I)$
\item[$\textsf{Hyb}_2$] In this hybrid, we introduce a simulator $\textsf{Sim}$ that has access to \emph{all} inputs $I$. The simulator runs a full simulation of the protocol $\pi$, and is thus indistinguishable from $\pi$.
\item[$\textsf{Hyb}_3$] In this hybrid, the simulator $\textsf{Sim}$ replaces the computed value $a$ (Protocol~\ref{prot:hdmm}, round 3, step 3) with the supplied output of the ideal functionality, and $\hat{M}_d$ and $\hat{M}$  (Protocol~\ref{prot:hdmm}, round 3, steps 1-2) with results computed from $a$ by inverting the computation of the protocol. These changes maintain indistinguishability by definition. The computation in round 3 is invertible because all operations are linear.
\item[$\textsf{Hyb}_4$] In this hybrid, the simulator $\textsf{Sim}$ replaces the inputs to $\mathcal{F}_\text{agg}$ for corrupted clients (Protocol~\ref{prot:hdmm}, round 2, step 5) with random values consistent with the output of $\mathcal{F}_\text{agg}$. Specifically, $\textsf{Sim}$ generates $\{\hat{m}_i \mid i \in C\}$ uniformly at random such that $\sum_{c_i \in U} \hat{m}_i + \sum_{c_i \in C} \hat{m}_i = \hat{M}$.
\end{description}
The distribution of the last hybrid can be computed with the simulator's inputs $\{I_i \mid i \in C\}$ and $a$.
\end{proof}

\subsection{Malicious Security}
\label{sec:maliciou_security}

As described in Section~\ref{sec:threat-models}, single-server secure aggregation protocols typically ensure confidentiality, but not correctness, in the presence of a malicious adversary. The protocol $\Pi_{\textsc{HDMM}}$, by instantiating $\mathcal{F}_{\text{agg}}$ with a malicious-secure aggregation protocol, also ensures confidentiality (but not correctness) in the presence of a malicious adversary. The argument is largely the same as the proof of Theorem~\ref{thm:semi-honest-security}, since the confidentiality of the protocol rests primarily on the confidentiality provided by $\mathcal{F}_\text{agg}$.

\paragraph{Corrupt clients.}
Corrupt clients may compute their inputs to $\mathcal{F}_\text{agg}$ incorrectly (Protocol~\ref{prot:hdmm}, round 2, step 5), including by adjusting or eliminating the noise added during the encoding step. These deviations will result in incorrect results from the protocol, but do not harm confidentiality of honest clients' inputs, as long as $\mathcal{F}_\text{agg}$ ensures confidentiality in the presence of a malicious adversary.

\paragraph{Corrupt server.}
A corrupt server may produce an incorrect strategy matrix $A$ (Protocol~\ref{prot:hdmm}, round 1) or compute the final output $a$ incorrectly (Protocol~\ref{prot:hdmm}, round 3). Either of these deviations will result in incorrect output from the protocol, but will not harm confidentiality of honest clients' inputs. Modifying $A$ may change the sensitivity of the clients' measurements, but clients compute this sensitivity locally, so the output of $\mathcal{F}_\text{agg}$ will satisfy differential privacy even when the strategy matrix is chosen maliciously.

\section{Evaluation}

We empirically evaluate Distributed HDMM to assess its scalability, efficiency, and accuracy in realistic federated settings. Using workloads derived from the U.S.\ Census SF1 and Adult datasets, we show that Distributed HDMM scales efficiently to tens of thousands of clients, introduces only modest computational and communication overhead, and achieves accuracy nearly indistinguishable from central-model HDMM while significantly outperforming local and shuffle models. 
%
Our experiments are designed to answer three key research questions:
\begin{itemize}[leftmargin=14pt]
    \item \textbf{RQ1}: Can Distributed HDMM can scale to thousands of clients while maintaining practical runtime?
    \item \textbf{RQ2}: How do network restrictions such as bandwidth and latency affect performance?
    \item \textbf{RQ3}: How does the protocol impact the accuracy of query answers compared to central, local, and shuffle-model baselines?
\end{itemize}
This section begins by describing our experiment setup; the remaining subsections describe results and answers to these research questions.

\begin{figure}
\centering
\includegraphics[width=0.5\textwidth]{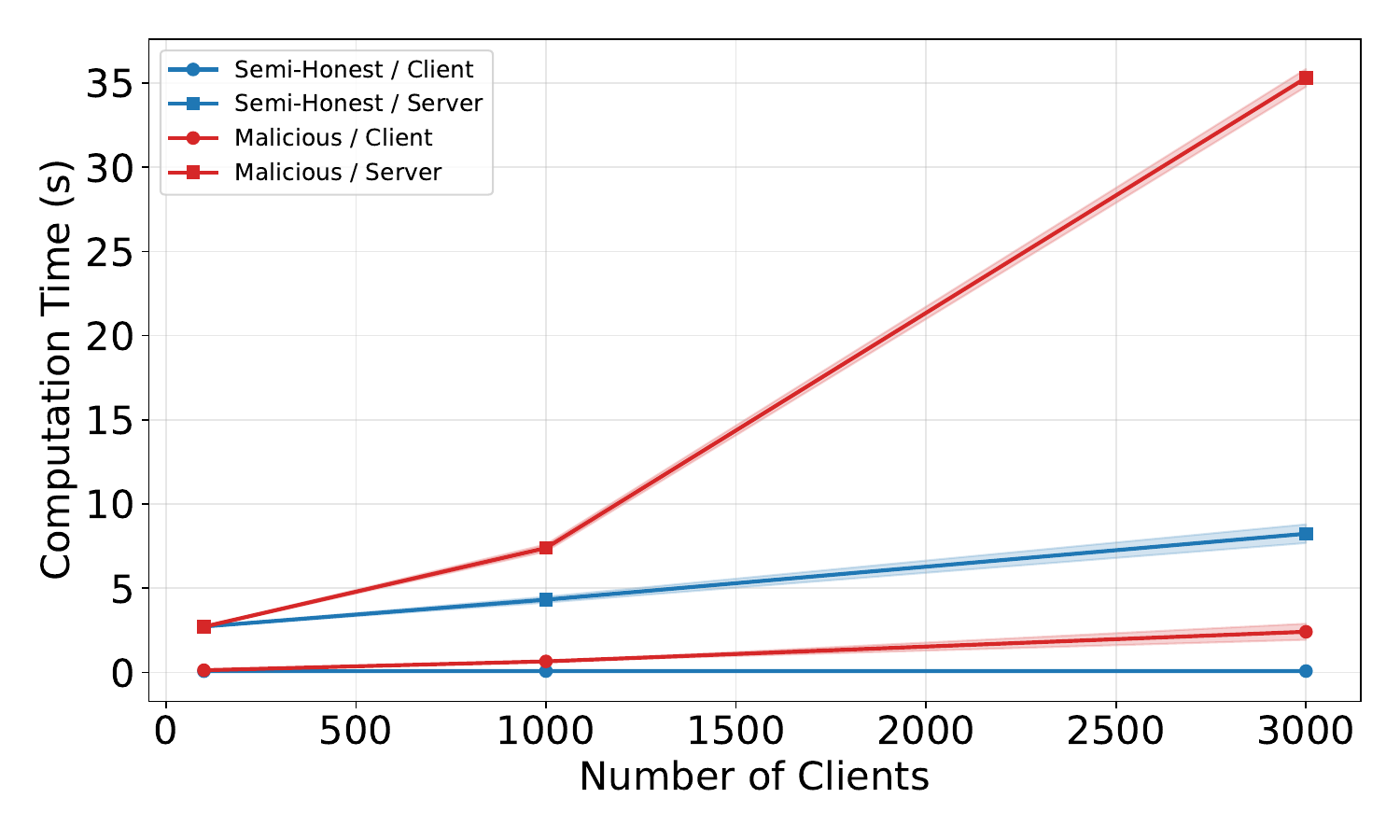} 
\caption{Server and Client (Average) Computation Time (Bandwidth = Unlimited)}
\label{fig:client-comp-time-sbw-unlimited}
\end{figure}







\begin{figure}
\centering
\includegraphics[width=0.5\textwidth]{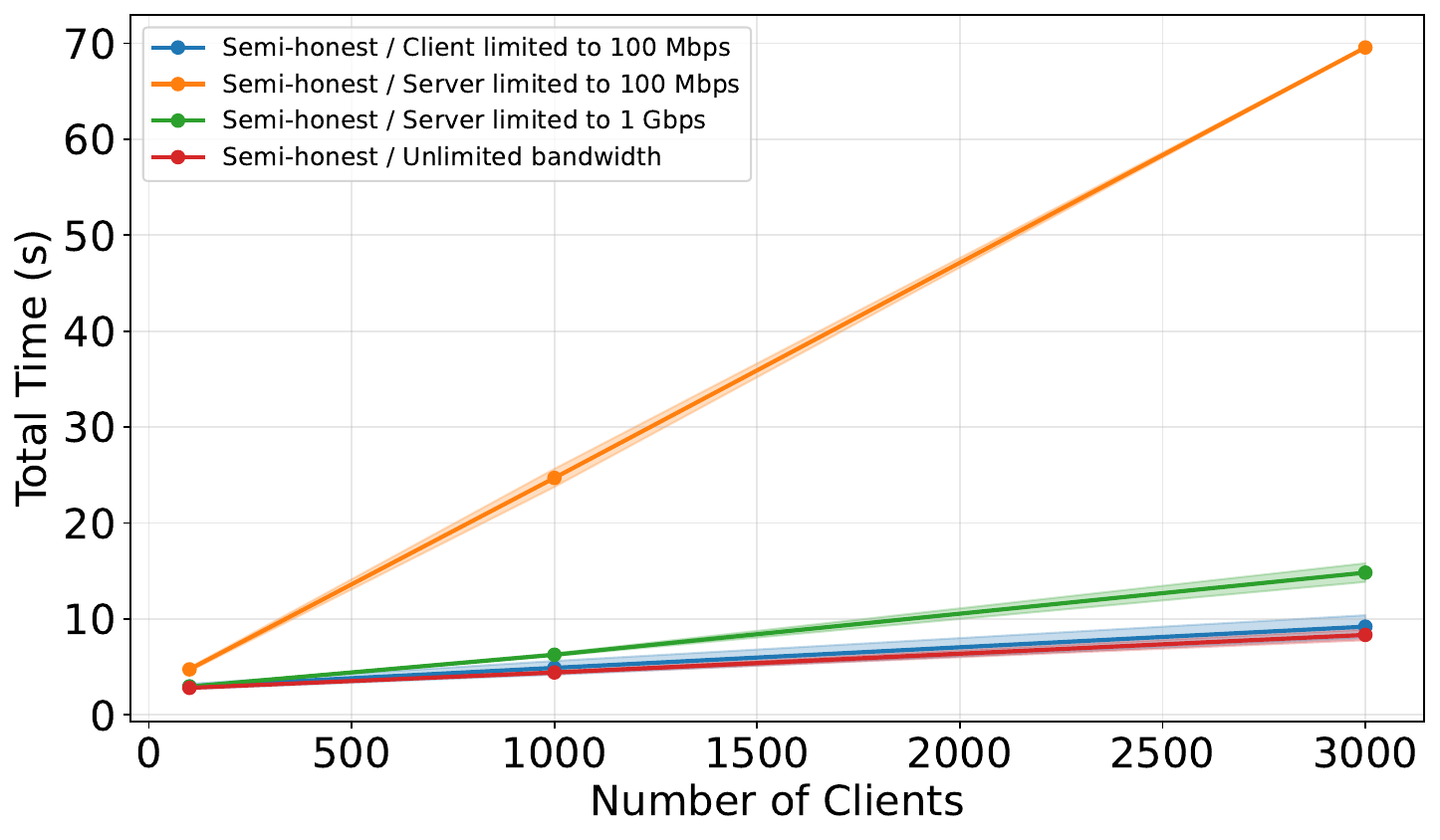} 
\caption{Total Protocol Running Time (s), semi-honest security}
\label{fig:total-time-semihonest}
\end{figure}

\begin{figure}
\centering
\includegraphics[width=0.5\textwidth]{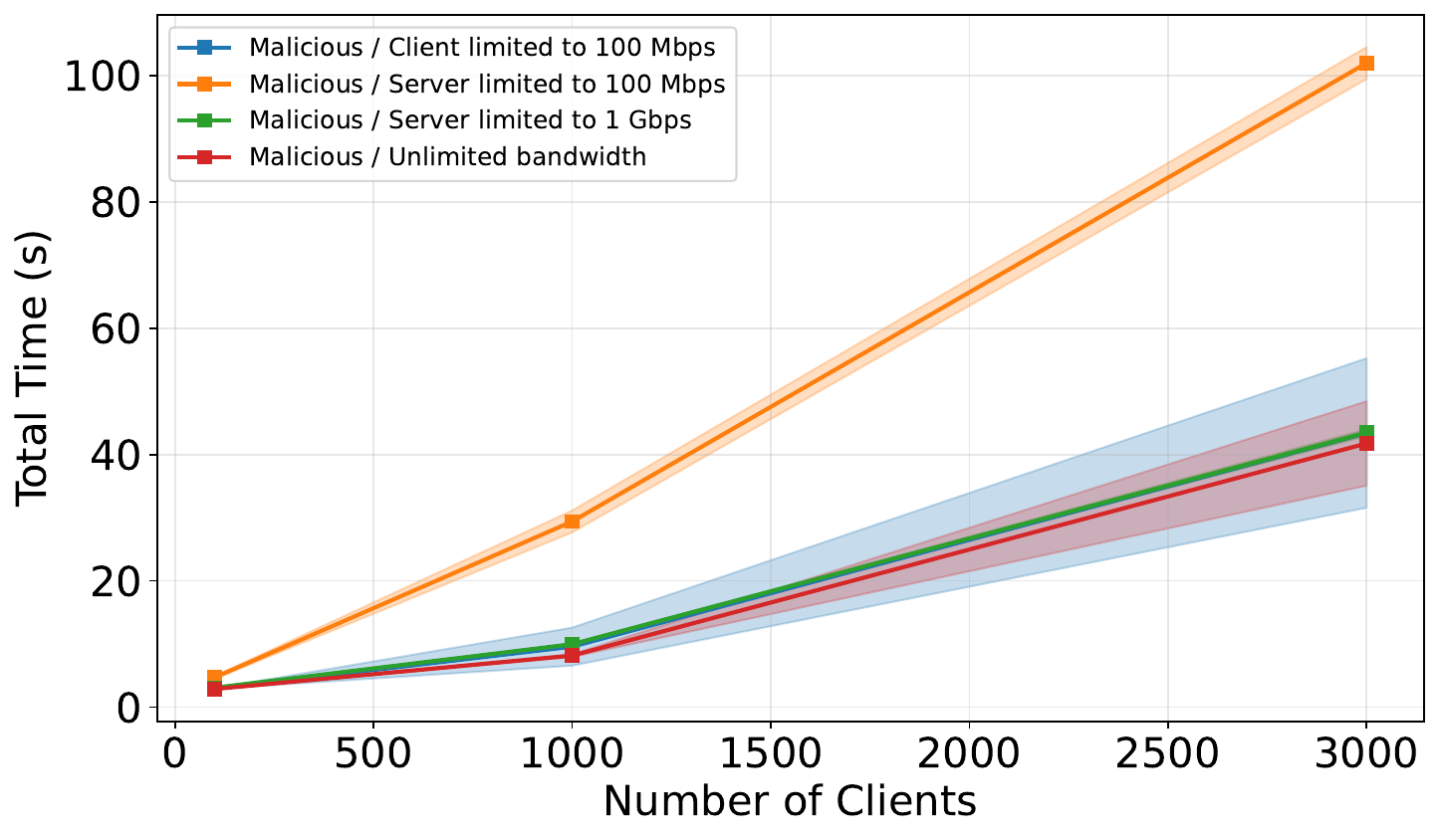} 
\caption{Total Protocol Running Time (s), malicious security}
\label{fig:total-time-malicious}
\end{figure}




\begin{figure}
\centering
\includegraphics[width=0.5\textwidth]{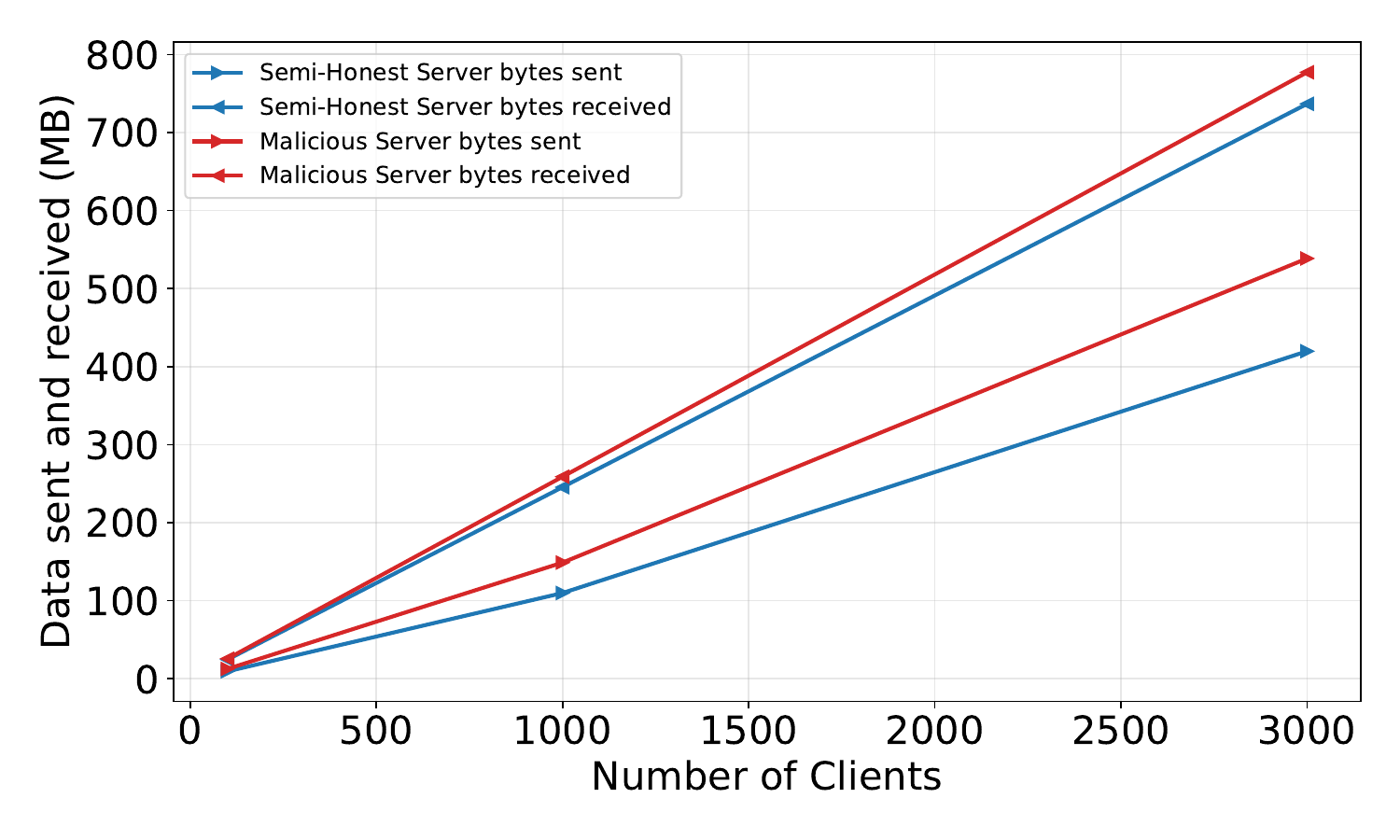} 
\caption{Server Communication Cost: sent and received traffic, in megabytes (MB)}
\label{fig:server-bytes}
\end{figure}

\begin{figure}
\centering
\includegraphics[width=0.5\textwidth]{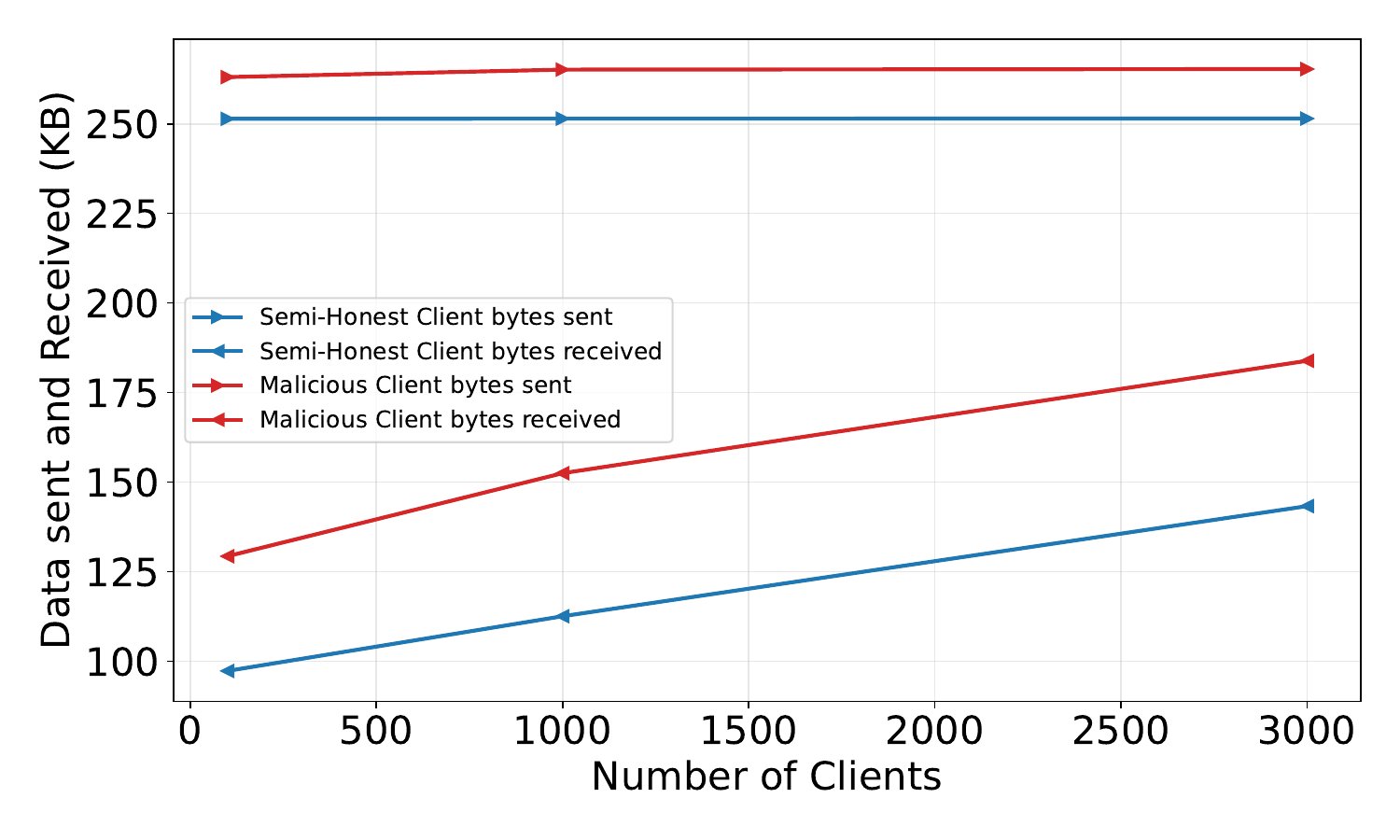} 
\caption{Average Client Communication Cost: sent and received traffic, in kilobytes (KB)}
\label{fig:client-bytes}
\end{figure}


\subsection{Implementation \& Experiment Setup}

Our empirical evaluation is designed to determine whether Distributed HDMM can scale to realistic workloads with data distributed across thousands of clients. We implemented the protocol using the \textsc{Olympia} simulation framework for secure protocols~\cite{ngong2023olympia}, which enables controlled experiments with thousands of simulated clients on a single physical machine. To represent query workloads and perform workload optimization, we used the existing Python implementation of the central-model HDMM released by McKenna et al.~\cite{hdmm_impl}. For secure aggregation, we instantiated $\mathcal{F}_{\text{agg}}$ with the Bell et al.~\cite{bell2020secure} protocol, which is communication-efficient and provides malicious security under an honest majority assumption. 

We evaluated Distributed HDMM on workloads derived from the U.S.\ Census Bureau Summary File 1 (SF1). The SF1 workload contains thousands of linear counting queries over high-dimensional contingency tables, and has been widely used in the HDMM literature as a benchmark. To simulate realistic federated data, we partitioned the input randomly across clients, with each client holding a subset of records. For scalability experiments, we configured runs with 100, 1{,}000, and 3{,}000 clients. Each experiment was repeated five times, and we report the mean and standard deviation of runtime. 

All experiments were conducted on a Linux server with 20 CPU cores and 128 GB of memory. We recorded (i) average client computation time, (ii) server computation time, and (iii) total end-to-end runtime of the protocol. 




\begin{table*}
\centering
\renewcommand{\arraystretch}{1.2}
\begin{tabular}{r c p{2.5cm} p{2.0cm} p{2.5cm} p{2.5cm} p{2.5cm}}
\hline
\textbf{Clients} & 
\textbf{Threat Model} & 
\textbf{Avg.\ Client Comp. Time (ms)} & 
\textbf{Server Comp. Time (s)} & 
\textbf{Total Runtime (s)} &
\textbf{Avg.\ Client Communication (KB)} &
\textbf{Server Communication (MB)} \\
\hline
100 & Semi-honest & 80ms  & 2.7s & 2.8s & 348 KB & 34 MB\\
1{,}000 & Semi-honest & 80ms & 4.1s  & 4.4s & 364 KB & 356 MB\\
3{,}000 & Semi-honest & 90ms  & 8.2s  & 8.3s & 395 KB & 1156 MB\\
\hline
100 & Malicious & 140ms  & 2.7s & 2.9s & 392 KB & 38.3 MB\\
1{,}000 & Malicious & 660ms  & 7.4s  & 8.2s & 418 KB & 408 MB\\
3{,}000 & Malicious & 2400ms  & 35.3s  & 41.8s & 449 KB& 1315 MB\\
\hline
\end{tabular}
\caption{Runtime results for Distributed HDMM under the Census SF1 workload. 
Values are averaged over five runs.}
\label{tbl:experiment_table}
\end{table*}




\subsection{RQ1: Scalability}

\paragraph{Experiment setup.}  
To evaluate scalability, we ran Distributed HDMM with 100, 1{,}000, and 3{,}000 simulated clients on the Census SF1 workload. Each client held a disjoint subset of randomly generated records, and we measured total runtime, client computation time, and server computation time under unlimited bandwidth and latency. Each experiment was repeated five times, and we report averages.

\paragraph{Distributed HDMM scales to thousands of clients.}  
Figures~\ref{fig:total-time-semihonest} and~\ref{fig:total-time-malicious} and Table~\ref{tbl:experiment_table} show that total runtime remains low even as the number of clients increases. Moving from 1{,}000 to 3{,}000 clients increases runtime by less than five seconds, confirming that the protocol scales logarithmically with the number of participants. This is expected: HDMM’s cost is dominated by workload optimization, which is independent of dataset size, while the Bell aggregation protocol is communication-efficient and scales well in the number of clients. Importantly, both the semi-honest and malicious settings exhibit scalability: in the semi-honest case, runtime grows slowly with the number of clients, remaining practical even for several thousand participants; in the malicious case, the additional verification and cryptographic checks introduce higher overhead, so scaling is somewhat less efficient, but still well within feasible limits for large deployments. These results indicate that Distributed HDMM is suitable for large federated deployments involving tens of thousands of clients under both adversarial models.

\paragraph{Distributed HDMM does not impose large communication costs.}  
Although not shown in a figure here, our measurements of bytes sent and received confirm that communication costs remain modest, as shown in \ref{fig:client-bytes} and ~\ref{fig:server-bytes}. Each client transmits only its masked measurement vector, whose size is proportional to the query workload but independent of the number of clients. As a result, communication grows linearly with the dimensionality of the workload rather than the number of participants. In the SF1 workload, the per-client communication was only a few megabytes, even with 3{,}000 clients. This demonstrates that communication overhead does not pose a scalability bottleneck.

\paragraph{Distributed HDMM requires modest client computation and reasonable server computation.}  
Figure~\ref{fig:client-comp-time-sbw-unlimited} shows that average client computation time remains low across all settings; it remains below 1 second in the semi-honest setting, and below 5 seconds in the malicious setting. This cost reflects a single matrix–vector multiplication, noise addition, and encoding step per client. In the malicious setting, the Bell protocol also requires additional verification of signed values, increasing computation cost. Because these operations are lightweight, Distributed HDMM places minimal computational burden on clients, making it practical even for resource-constrained devices in federated environments. On the server side, computation is somewhat higher in both semi-honest and malicious cases, driven primarily by HDMM’s workload optimization step and the cryptographic operations required to process aggregated inputs. As shown in Figure~\ref{fig:client-comp-time-sbw-unlimited}, this overhead remains modest, with total server time well under a few seconds even at 3,000 clients, though malicious-secure execution incurs additional cryptographic cost relative to the semi-honest setting.

\subsection{RQ2: Impact of Network Restrictions}

\paragraph{Experiment setup.}  
To understand the effect of network constraints, we introduced artificial bandwidth and latency limits into our simulation. We considered two representative scenarios: (1) limiting clients to 1 Mbps upload bandwidth with 100 ms latency, and (2) limiting the server to 1000 Mbps bandwidth with 100 ms latency. These scenarios approximate realistic heterogeneous network conditions in federated deployments, such as mobile devices with limited uplink capacity and servers with constrained aggregation bandwidth. We measured total runtime under both configurations and compared the results to the unlimited bandwidth baseline.

\paragraph{Client bandwidth limitations do not have significant impact on scalability.}  
Figure~\ref{fig:client-comp-time-sbw-unlimited} shows that constraining clients to 1 Mbps has negligible effect on runtime. Even with 10{,}000 clients, total runtime remains nearly identical to the unlimited-bandwidth case. This is because each client transmits only a single masked measurement vector of modest size, and the Bell aggregation protocol does not require extensive interactive communication. These results suggest that Distributed HDMM is well-suited to federated environments with heterogeneous or bandwidth-limited clients, since communication cost per client is small and does not grow with the number of participants.

\paragraph{Server bandwidth limitations have significant impacts on scalability.}  
In contrast, Figure~\ref{fig:client-comp-time-sbw-unlimited} shows that restricting the server to 100 Mbps bandwidth increases runtime noticeably, particularly as the number of clients grows. Unlike clients, the server must receive and process vectors from all participants, so its communication load scales linearly in the number of clients. This makes server uplink capacity the primary bottleneck in very large deployments. Nevertheless, even under this restriction, the protocol completed within tens of seconds for up to 3{,}000 clients. These findings suggest that server provisioning—rather than client communication—will determine scalability in practice, and highlight the importance of high-capacity aggregation servers for production deployments.

\begin{figure*}
\centering
\hspace*{-10pt}\begin{tabular}{c c}
\includegraphics[width=.5\textwidth]{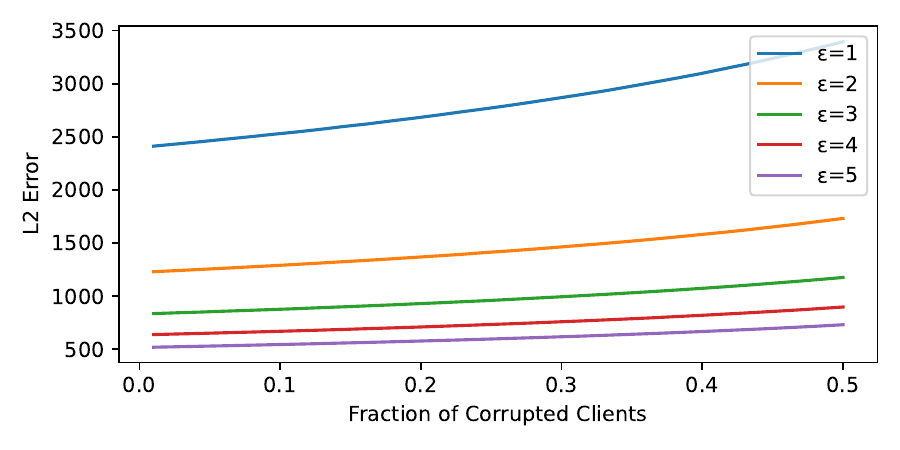}
&
\includegraphics[width=.5\textwidth]{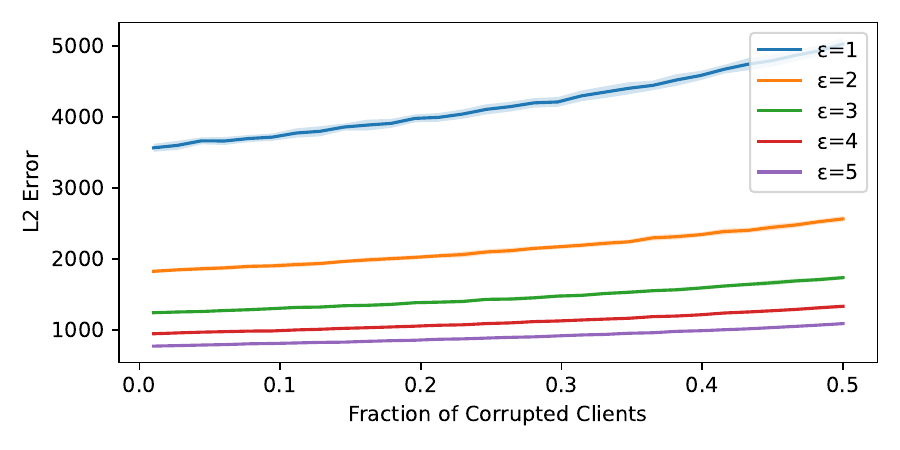}\\
\textbf{Adult} & \textbf{Census SF1}\\
\end{tabular}
\caption{Impact of Distributed HDMM on utility, for 2-way marginals on the Adult dataset and for the Census SF1 workload and 1000 clients. Distributed HDMM produces the same utility as central-model HDMM when $\theta=0$ (left edge of the graph). Error increases slowly with $\theta$.}
\end{figure*}

\begin{figure*}
\centering
\hspace*{-10pt}\begin{tabular}{c c}
\includegraphics[width=.5\textwidth]{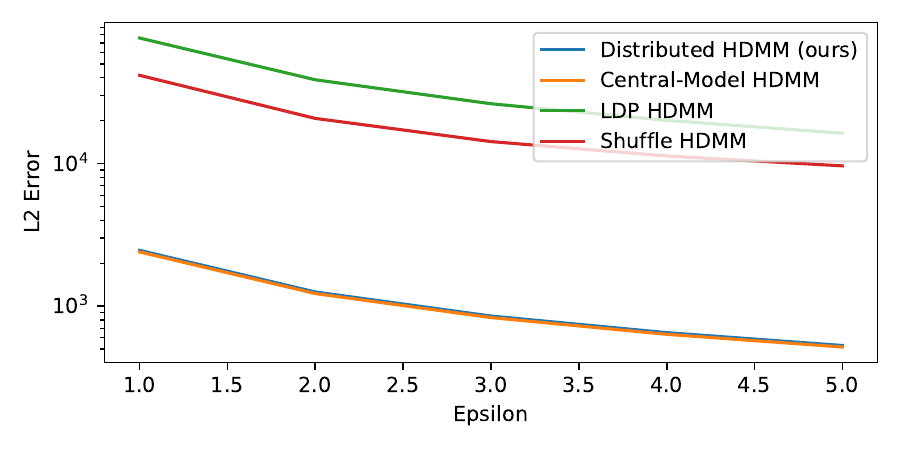}
&
\includegraphics[width=.5\textwidth]{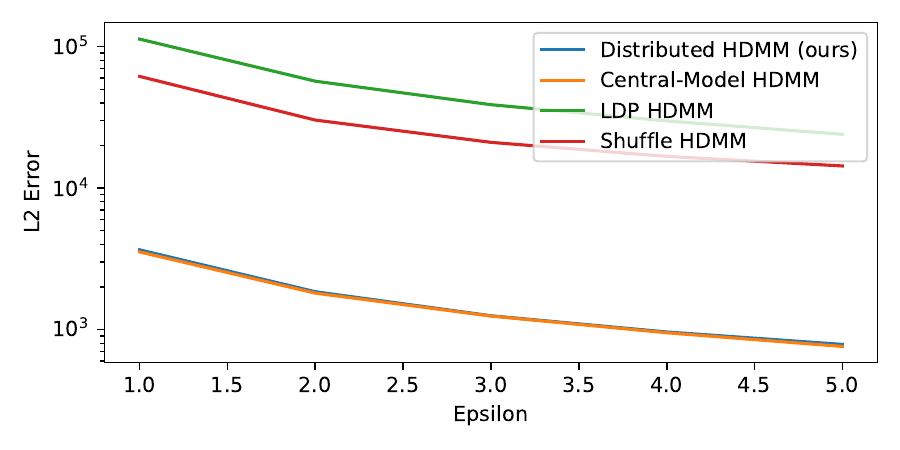}\\
\textbf{Adult} & \textbf{Census SF1}\\
\end{tabular}
\caption{Utility comparison of Distributed HDMM to local-model, central-model, and shuffle-model baselines, for $\theta=0.05$ and 1000. Note logarithmic vertical axis. Distributed HDMM nearly matches the utility of central-model HDMM, while both the local model and shuffle model increase error by an order of magnitude.}
\end{figure*}

\subsection{RQ3: Impact on Utility}

\paragraph{Experiment setup.}  
To evaluate accuracy, we compared Distributed HDMM to three baselines: the central-model HDMM, the local model (LDP), and the shuffle model. Following McKenna et al.~\cite{mckenna2018optimizing}, we measured error on two standard benchmarks: the Adult dataset (two-way marginals) and the Census SF1 workload. We varied the corruption fraction $\theta$ from 0 to 0.3 to examine robustness, and ran experiments at multiple privacy budgets ($\epsilon \in \{1, 2, 3, 4, 5\}$). Error was reported as $L_2$ error (i.e. root mean squared error (RMSE)) between the noisy and true query answers. We set the number of clients to 1000.

\paragraph{Distributed HDMM nearly matches central-model accuracy.}  
Figure~7 shows that when $\theta=0$, Distributed HDMM achieves the same accuracy as the central model, confirming that secure aggregation does not introduce additional error beyond that of HDMM itself. As $\theta$ increases, error grows slowly because honest clients must add slightly more noise to compensate for potentially corrupted participants. Even with $\theta=0.3$, the error remains within a small constant factor of the central-model baseline.

\paragraph{Comparison to local and shuffle models.}  
Figure~8 highlights the difference between Distributed HDMM and baselines representing alternative approaches. At $\theta=0.05$, Distributed HDMM nearly matches the central model, while both local and shuffle models increase error by roughly an order of magnitude. The shuffle model provides better accuracy than local DP, but still lags far behind Distributed HDMM. This confirms that Distributed HDMM closes the accuracy gap while retaining the trust assumptions of distributed protocols.
The cryptographic constructions used in Distributed HDMM do introduce additional computational and communication overhead compared to local model or shuffle model approaches, and this overhead grows (slowly) with the number of clients, a shown in Figures~\ref{fig:total-time-semihonest} and \ref{fig:total-time-malicious}. For most deployments, the improvement in utility from Distributed HDMM is likely to be worth the small increase in running time for the protocol.


\section{Related Work}

\paragraph{Distributed differentially private query answering via local differential privacy}  
Local differential privacy (LDP)~\cite{erlingsson2014rappor, ding2017collecting, wang2017locally} eliminates the need for a trusted curator by having each client perturb its data locally before sending it to the server. While attractive for its minimal trust assumptions, LDP typically suffers from high variance and poor accuracy, especially on high-dimensional workloads. For example, frequency estimation and marginal queries under LDP require far more noise than central-model approaches, resulting in error an order of magnitude larger than central DP in many practical settings. Our work avoids this limitation by combining secure aggregation with HDMM, achieving utility close to the central model while maintaining distributed trust.

\paragraph{Distributed differentially private query answering via shuffle differential privacy}  
The shuffle model~\cite{cheu2019distributed, balle2019privacy, erlingsson2019amplification} strengthens LDP by introducing an additional non-colluding shuffler that randomly permutes clients’ messages before forwarding them to the server. This provides improved privacy amplification, often narrowing the accuracy gap with the central model. However, the shuffle model requires trust in a non-colluding shuffler, and confidentiality collapses if the server and shuffler collude. In addition, the amount of noise required in shuffle protocols generally scales with the number of honest clients, meaning utility can degrade significantly under adversarial participation (see Table 3 in~\cite{cheu2021differential}). By contrast, Distributed HDMM tolerates a malicious server, requires no shuffler, and retains near-central accuracy even when up to half of the clients are corrupted.

\paragraph{Distributed differentially private query answering via MPC}  
Another line of work uses secure multiparty computation (MPC) to realize central-model accuracy without a trusted curator. Shrinkwrap~\cite{bater2018shrinkwrap}, Crypt$\epsilon$~\cite{roy2020crypt}, and IncShrink~\cite{wang2022incshrink} achieve central-model DP through MPC protocols, but rely on two non-colluding semi-honest servers, with clients not participating directly. Honeycrisp~\cite{roth2019honeycrisp} and Orchard~\cite{roth2020orchard} introduce specialized secure aggregation protocols that scale to millions of participants, but depend on small committees and assume semi-honest aggregation servers.  

Each of these approaches differs fundamentally from ours. Shrinkwrap is designed for private data federations and focuses on minimizing padding in oblivious query processing, not on federated-scale client participation. Crypt$\epsilon$ bridges local and central models but requires two non-colluding servers, which our protocol does not. IncShrink incrementally constructs differentially private synopses, but permits untrusted servers to view DP summaries directly. Honeycrisp and Orchard rely on committee selection and tolerate only a small fraction of corrupted clients, whereas our protocol tolerates up to half of clients being adversarial and does not require committees or multiple servers.  

LDP and shuffle approaches sacrifice accuracy for minimal trust, while MPC-based approaches often require additional trust assumptions such as non-colluding servers or committees. Distributed HDMM is the first to combine the accuracy of the central-model HDMM with the scalability and robustness of single-server secure aggregation, achieving utility within a small constant factor of central DP while requiring only an honest majority of clients.


\section{Conclusion}    
We introduced \emph{Distributed HDMM}, a protocol that brings the accuracy benefits of the central-model High-Dimensional Matrix Mechanism into a distributed setting without relying on a trusted curator. By combining optimized strategy matrices with secure aggregation, our approach achieves accuracy within a small constant factor of the central model while maintaining strong privacy and robustness guarantees against both semi-honest and malicious adversaries.  

Our evaluation demonstrates that Distributed HDMM scales to thousands of clients, incurs modest computational overhead, and consistently delivers accuracy that far surpasses local and shuffle-model approaches. These results highlight that high-utility differentially private data analysis is feasible even in federated or multi-institutional environments where no fully trusted curator exists.  


\bibliographystyle{plain}
\bibliography{refs,sec}
\end{document}